\documentclass[twocolumn]{revtex4-1}

\usepackage{latexsym}
\usepackage{amssymb}
\usepackage{amsmath}
\usepackage{amsthm}
\usepackage{amsfonts}
\usepackage{ifthen}
\usepackage{revsymb}
\usepackage{yfonts}
\usepackage{graphicx}
\usepackage{algpseudocode}
\usepackage{bbm}
\usepackage{float}
\usepackage{multirow}
\usepackage{gensymb}
\usepackage{epstopdf}
\usepackage{listings}
\usepackage{color}
\definecolor{mygreen}{rgb}{0,0.6,0}
\definecolor{mygray}{rgb}{0.5,0.5,0.5}
\definecolor{mymauve}{rgb}{0.58,0,0.82}

\usepackage{afterpage}
\usepackage{verbatim}
\usepackage{soul}
\usepackage{lineno, blindtext}
\usepackage{footnote}
\usepackage{afterpage}
\usepackage{algpseudocode}
\usepackage[normalem]{ulem}
\usepackage{algorithm}
\usepackage{algpseudocode}
\usepackage{mathtools}

\usepackage[colorlinks = true]{hyperref}
\usepackage{longtable}
\usepackage{xcolor}
\definecolor{darkred}  {rgb}{0.5,0,0}
\definecolor{darkblue} {rgb}{0,0,0.5}
\definecolor{darkgreen}{rgb}{0,0.5,0}
\usepackage{amsmath}

\hypersetup{
  urlcolor   = blue,         
  linkcolor  = darkblue,     
  citecolor  = darkgreen,    
  filecolor  = darkred       
}

\newcommand{\be}{\begin{equation}}
\newcommand{\ee}{\end{equation}}
\newcommand{\bq}{\begin{eqnarray}}
\newcommand{\eq}{\end{eqnarray}}
\newcommand{\bea}{\begin{eqnarray}}
\newcommand{\eea}{\end{eqnarray}}
\newcommand{\ba}{\begin{align}}
\newcommand{\ea}{\end{align}}

\newcommand{\ket}[1]{ | \, #1 \rangle}

\newtheorem{theorem}{Theorem}

\newtheorem{proposition}[theorem]{Proposition}

\definecolor{mygray}{gray}{0.6}

\definecolor{mygray}{gray}{0.9}

\begin{document}
\title{A Novel Quantum N-Queens Solver Algorithm and its Simulation and Application to Satellite Communication Using IBM Quantum Experience}

\newcommand{\IISERP}{Department of Physical Sciences, Indian Institute of Science Education and Research, Pune, 411021, India}

\newcommand{\NISERJ}{School of Physical Sciences, National Institute of Science Education and Research, HBNI, Jatni, 752050, India}
\newcommand{\NMITB}{Department of Computer Science and Engineering, Nitte Meenakshi Institute of Technology, Bangalore, 560054, India}

\newcommand{\IISERK}{Department of Physical Sciences, Indian Institute of Science Education and Research Kolkata, Mohanpur, 741246, West Bengal, India}

\author{Rounak Jha$^1$}
\thanks{These authors contributed equally to this work.}
\author{Debaiudh Das$^2$}
\thanks{These authors contributed equally to this work.}
\author{Avinash Dash$^3$}
\thanks{These authors contributed equally to this work.}
\author{Sandhya Jayaraman$^4$}
\author{Bikash K. Behera$^3$}
\author{Prasanta K. Panigrahi$^3$}
\affiliation{$^1$\IISERP}
\affiliation{$^2$\NISERJ}
\affiliation{$^3$\IISERK}
\affiliation{$^4$\NMITB}

\begin{abstract}
      Quantum computers can potentially solve problems that are computationally intractable on a classical computer in polynomial time using quantum-mechanical effects such as superposition and entanglement. The $N$-Queens Problem is a notable example that falls under the class of NP-complete problems. It involves the arrangement of $N$ chess queens on an $N \times N$ chessboard such that no queen attacks any other queen, i.e. no two queens are placed along the same row, column or diagonal. The best time complexity that a classical computer has achieved so far in generating all solutions of the $N$-Queens Problem is of the order $\mathcal{O}(N!)$. In this paper, we propose a new algorithm to generate all solutions to the $N$-Queens Problem for a given $N$ in polynomial time of order $\mathcal{O}(N^3)$ and polynomial memory of order $\mathcal{O}(N^2)$ on a quantum computer. We simulate the $4$-queens problem and demonstrate its application to satellite communication using IBM Quantum Experience platform. 
\end{abstract}
 
\maketitle
 
\section{Introduction}
NP complete problems encompass all decision problems in class NP that can be verified in polynomial time \cite{qnq_UllmanJCSC1975,qnq_FarhiScience2001}. However, the time taken to obtain a solution for an NP complete problem is not polynomial. Some common examples of NP complete problems include the traveling salesman problem \cite{qnq_KargMS1964}, the subset sum problem \cite{qnq_CapraraAJO2000}, the Hamiltonian cycle problem \cite{qnq_PlesnikIPL1979}, the knapsack problem  \cite{qnq_ChuJH1998,qnq_JaszkiewiczIEEE2002} and satisfiability problem \cite{qnq_SuIEEE2016,qnq_PudnezarXiv2016} to name a few. 

The $N$-Queens Problem is an NP complete problem \cite{qnq_TorgglerarXiv2018}. It states that $N$ queens must be placed on an $N\times N$ chessboard such that no two queens can attack each other \cite{qnq_BellStevens2009}. The queens follow the moves that a classical chess queen follows, i.e. horizontal, vertical and diagonal jumps across any number of squares as long as the queen is unobstructed by the presence of another queen. The generalized $N$-Queens Problem \cite{qnq_BellStevens2009} has already been attempted over the years \cite{qnq_SouzaIEEE2017}. It has been approached using multiple classical computational methods \cite{qnq_AlharbiJO2017}. Some of these include the brute force algorithm \cite{qnq_MukherjeeIJFCST2015}, different variants of the backtracking algorithm \cite{qnq_KhateebDSE2013,qnq_ThadaIJCA2014} and the greedy technique \cite{qnq_LijoIJCSIT2015}. The brute force algorithm has the highest growth rate (with $N$) of the order $\mathcal{O}(N^N)$. The backtracking algorithm is slightly better than the former with a time complexity of the order $\mathcal{O}(N!)$. Other variants of the backtracking method such as the optimized backtracking and novice backtracking algorithms have an exponential time complexity \cite{qnq_ZengIEEE2011}.~The systematic and greedy search methods have a time complexity of the order $\mathcal{O}(N^3)$ to $\mathcal{O}(N^2)$. These algorithms, although solvable in polynomial time, often yield a set of limited solutions to the $N$-Queens Problem for a given $N$.

The $N$-Queens Problem finds vast application in various fields such as parallel memory storage schemes \cite{qnq_ErbasParallelStorage1992}, low density parity check codes \cite{qnq_LiParity2004}, deadlock prevention \cite{qnq_BellStevens2009}, VLSI testing \cite{qnq_WilliamsACV1986}, neural networks \cite{qnq_CheSciRep2018,qnq_JiangIEEE2018} and load balancing \cite{qnq_PanwarLoadBalancing2013} etc. 

\section{Solving the N-Queens Problem}

The criteria for a particular configuration of an $N\times N$ chessboard to be a solution to the $N$-Queens Problem are stated below. We shall collectively refer to them as the $N$-Queens Criteria.

\begin{itemize}
\item \textbf{Row Criteria:} The total number of queens in each row is $1$.
\item \textbf{Column Criteria:} The total number of queens in each column is $1$.  
\item \textbf{Diagonal Criteria:} The total number of queens in every possible diagonal is either $0$ or $1$.
\end{itemize}

We formulate this problem by representing the $N \times N$ chessboard as an $N \times N$ matrix. The positions of queens and the vacant spaces are denoted by $1$s and $0$s respectively. We represent each row of the resultant matrix as an $N$-qubit basis state. Hence, the $N$-qubit quantum register for all possible row configurations having one queen can be represented as,
\begin{equation}
    \ket{\psi}=\frac{1}{\sqrt{N}}[\ket{R_1}+\ket{R_2}+....+\ket{R_{N-1}}+\ket{R_N}]
    \label{qnq_nq_1}
\end{equation}

where $R_i$ ($i\in [1,N]\cap\mathbb{N}$) represents the $i$th row of the $N \times N$ identity matrix. The above state ($\ket{\psi}$) is known as the $N$-qubit $W$ state. Since there are $N$ such rows, we can create $N$ such $W$ states $\{\ket{\psi_i}\}$, where $\ket{\psi_i}$ represents the $W$ state corresponding to the $i$th row. We take the tensor product of all $\ket{\psi_i}$ as follows,
\begin{equation}
    \ket{\Psi}=\ket{\psi_1}\otimes\ket{\psi_2}\otimes....\otimes\ket{\psi_N}
\label{qnq_nq_2}
\end{equation}

The $N^2$-qubit state $\ket{\Psi}$ is a superposition of $N^N$ basis states, where each basis state represents a particular configuration of the $N$ queens on the $N \times N$ chessboard, such that no two queens are in the same row. Thus, by forming $\ket{\Psi}$, we have generated a set of configurations which satisfy the Row Criteria as required. The communication channel prepares the $N^2$-qubit system in $\ket{\Psi}$ state \cite{qnq_zheng2004,qnq_deng2006,qnq_grafe2014chip} before the quantum computer implements the proposed algorithm.

The first computational task is to distinguish those basis states in $\ket{\Psi}$ which satisfy the Column Criteria. For this purpose, it is convenient to visualize the system of $N^2$ qubits as $N$-blocks comprising $N$ qubits each, where the $i$th block represents the $i$th row of the matrix. In this step, we figure out an optimal algorithm to check whether there is only one queen in each column or not. Consider any particular column. We take an ancillary qubit initialized in the state $\ket{0}$. We perform the Hadamard operation on this qubit which transforms it to $\ket{+}$ state. We then apply a controlled phase shift $e^{i\pi}$ to each of the corresponding qubits specifying the column, with the control being the ancilla. For example, if we are considering the first column, we would apply the given operation to the first qubits of each block. As a result, the ancillary qubit transforms to $\ket{-}$ state if the column sum is odd, otherwise it remains in $\ket{+}$ state. We reapply the Hadamard gate to the resultant state. As a result, the ancilla is in the state $\ket{1}$ if the number of queens in that column is odd, and in the state $\ket{0}$ otherwise.

Since each block of the qubits is in the $W$ state, we have the constraint that in each block only a single qubit is in $\ket{1}$ state, which in turn implies that the system of $N^2$ qubits has $N$ qubits in $\ket{1}$ state corresponding to $N$ queens being placed on the board. Hence, if there exists a column whose sum is $m$, such that $m$ is odd and $m>1$, then there exists at least one column whose sum is zero, since the column sums add up to $N$. As a result, the ancillary qubit corresponding to such a column remains in $\ket{0}$ state. Thus, the only possibility in which all the ancillas end up in $\ket{1}$ state is when each column sum is $1$ as followed by the Column Criteria. 

However, it can be observed that, if $N$ non-negative integers add up to $N$, then there must be an even number of even numbers among them, where zero is also considered as an even number. A proof of this is presented in Supplementary Section. This enables us to perform the column checks for only $N-1$ columns instead of all $N$ columns. In our algorithm, we perform the column check operations for all columns except the last one, hence requiring $N-1$ additional qubits in total.

As a result of the column check operations, the $N-1$ ancillas are entangled to the system of $N^2$ qubits. It is to be noted that, all the entangled ancillas will be in the $\ket{1}$ state if and only if the associated $N^2$-qubit basis states are representative of the configuration satisfying the Column Criteria, i.e., they encode matrices which can be formed by row permutations of the $N \times N$ identity matrix. The row and column checks together have thus reduced the size of the search space to $N!$. The corresponding quantum circuit for implementing the column checks for any $N$ has been presented in Fig.~\ref{qnq_Fig1} \textbf{Part I}.

\begin{figure*}
    \includegraphics[width=\linewidth]{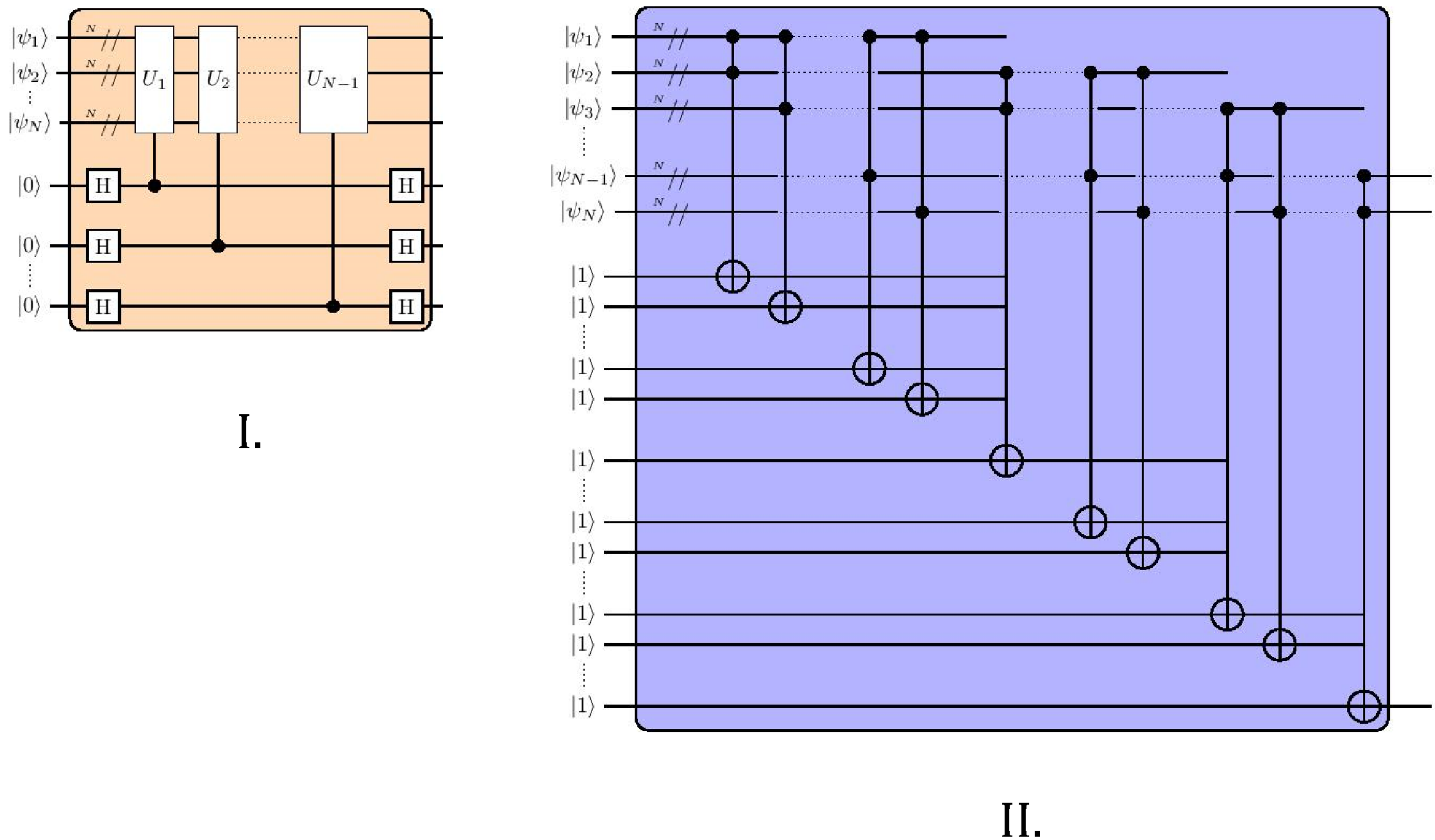}
    \caption{\textbf{Generalized quantum circuit for N-Queens Solver.} \textbf{Part I.} The generalized quantum circuit performs the column check operation in order to distinguish the basis states of $\ket{\Psi}$ satisfying the Column Criteria for reducing the search space. $N-1$ ancillary qubits along with a series of $N-1$ controlled unitary operations are used where the ancillas act as the controls. The operation $U_i$ denotes a set of $\sigma_z$ operations on the $i$th qubit of all the $N$ blocks, \textit{viz.} $\ket{\psi_1}$, $\ket{\psi_2}$,...., $\ket{\psi_N}$. Thus, each controlled unitary operation in turn represents a series of $N$ controlled-$Z$ operations. The set of $i$th qubits from every block represents the $i$th column of the chessboard. The $i$th ancillary qubit flips to $\ket{1}$ if the $i$th column sum is odd. Since the system of qubits satisfies the Row Criteria, the total sum of all $N$ column sums is $N$. As a result, all the $N-1$ ancillas flip to $\ket{1}$ if and only if each column sum is unity, which is only possible for basis states satisfying the Column Criteria. The circuit entangles the $N-1$ ancillas to the system. The reduced search space consists of basis states with all the $N-1$ entangled qubits in $\ket{1}$ state. The reduced search space is the set of configurations which are row permutations of the $N\times N$ identity matrix. \textbf{Part II.} The generalized quantum circuit performs the diagonal check on the basis states of $\ket{\Psi}$ belonging to the reduced search space, following the column check operation. A total of $\frac{N^2-N}{2}$ additional qubits are required for this purpose. A series of $3$-qubit operations are performed. In the circuit, the first control block of any $3$-qubit operation, say $\ket{\psi_i}$, is the one encoding $Q_i$, i.e. the queen in the $i$th row of the chessboard. The second control block of the $3$-qubit operation, say $\ket{\psi_j}$ encodes the matrix elements in the $j$th row of the chessboard which are diagonal to $Q_i$. Since for different configurations from the reduced search space, the queen $Q_i$ can be present in any one out of $N$ columns for any given row $i$, the qubit belonging to block $\ket{\psi_i}$ encoding $Q_i$ could be any one out of the $N$ qubits of the block, say $\ket{q_{ix}}$ for some $x\in \{1,2,....,N\}$. Here, each $3$-qubit operation represents a set of Toffoli gates, such that a $3$-qubit operation with the first and second control blocks being $\ket{\psi_i}$ and $\ket{\psi_j}$ respectively denotes a set of Toffoli gates with the first control qubit belonging to the set $\{\ket{q_{ix}}:\text{ }x\in[1,N]\cap\mathbb{N}\}$ and the second control qubit belonging to the set $\{\ket{q_{jy}}\}$ for a given $\ket{q_{ix}}$, where $\{\ket{q_{jy}}\}$ is the set of qubits in block $\ket{\psi_j}$ encoding matrix elements (in the $j$th row) which are diagonal to $\ket{q_{ix}}$. This leads to entanglement of the ancillas to the system. The basis states satisfying the Diagonal Criteria will not flip any auxiliary qubit. Hence, all the $\frac{N^2-N}{2}$ entangled qubits are in $\ket{1}$ state. Overall, the result of this circuit is that only the basis states which satisfy the $N$-Queens Criteria will have all the entangled $(\frac{N^2}{2}+\frac{N}{2}-1)$ ancillas in $\ket{1}$ state, and hence can be distinguished from the other basis states.}
	\label{qnq_Fig1}
\end{figure*}

Our final task is to extract the basis states satisfying the Diagonal Criteria. Now that our search space has been reduced to the row permutations of $I_{N \times N}$, there is only one queen in each row and each column. Let $Q_i$ denote the queen placed in the $i$th row. We define a function $f$ as follows.

\begin{eqnarray}
    &&\text{Let }Q=\{Q_1,Q_2,....,Q_N\}\nonumber \\
    &&f:Q\times Q \rightarrow \{0,1\} \text{ defined by, }\nonumber \\
    &&f(Q_i,Q_j)= \begin{cases}
       0\text{, if }j>i \text{ and }Q_i \text{, } Q_j\text{ are along a diagonal}\nonumber \\ 
       1\text{, otherwise}
       \end{cases}
    \label{qnq_nq_3}
\end{eqnarray}

It can be easily seen that, if for a particular configuration $f(Q_i,Q_j)=1$ $ \forall (Q_i,Q_j)\in Q \times Q$, then the configuration is a solution to the $N$-Queens Problem. We now propose the quantum algorithm to perform the diagonal checks to separate out the basis states of $\ket{\Psi}$ satisfying the $N$-Queens Criteria from the reduced search space of $N!$ basis states (which have been ``marked" with the $(N-1)$-qubit in $\ket{1}$ state of the entangled ancillary qubits, following the column check). Like earlier, it is convenient to imagine the system as a composition of $N$ blocks having $N$ qubits each, with the $i$th block corresponding to the $i$th row of the chessboard. Hence, the $y$th qubit of the $x$th block corresponds to the element of the $x$th row and $y$th column of the matrix. Let us denote such a qubit as $\ket{q_{xy}}$. Since the search space is a superposition of all basis states corresponding to configurations which are permutations of $I_{N\times N}$, any queen $Q_i$ corresponds to a $\ket{1}$ in the $i$th block is encoded by the qubit $\ket{q_{iy}}$ for some $y \in \{1,2,....,N\}$. It is observed that $|A|$ equals $\frac{N(N-1)}{2}$, where $A=\{(Q_i,Q_j)\in Q\times Q:j>i\}$. As a result, we require $\frac{N(N-1)}{2}$ additional qubits to calculate all possible evaluations of $f(Q_i,Q_j)$ $\forall (i,j):\text{ }j>i$. For a particular $i\in [1,N]\cap\mathbb{N}$, $(N-i)$ qubits out of the $\frac{N(N-1)}{2}$ ancillas are used to evaluate $f(Q_i,Q_j) \forall$ $j>i$. All the ancillas (denoted as $\ket{\delta_1}$, $\ket{\delta_2}$,...., $\ket{\delta_{\frac{N^2-N}{2}}}$) are initialized in $\ket{1}$ state. Let $C^2NOT$ denotes a Toffoli gate with the first and second qubits being the control and the third being the target. We perform the following set of operations $\{C^2NOT\ket{q_{ix}}\ket{q_{jy}}\ket{\delta_k}:i,j,x,y\in[1,N]\cap\mathbb{N},j>i, k=\frac{(i-1)}{2}(2N-i)+(j-i)\text{ and }q_{ix},q_{jy}\text{ represent matrix elements along a diagonal}\}$. The result of this set of operations is that, $\ket{\delta_k}$ flips to $\ket{0}$ if and only if $\ket{q_{ix}}$ and $\ket{q_{jy}}$ are in $\ket{1}$ state, which is equivalent to evaluating $f(Q_i,Q_j)$ for all configurations of the search space. As a result, all the $\frac{N(N-1)}{2}$ ancillas are entangled to the system $\ket{\Psi}$. The basis states that represent permutations of $I_{N\times N}$ and also satisfy the Diagonal Criteria will have all $\frac{N(N-1)}{2}$ associated ancillas in $\ket{1}$ state. These basis states encode the solutions to the $N$-Queens Problem. The corresponding circuit for implementing the diagonal checks for any $N$ has been presented in Fig.~\ref{qnq_Fig1} \textbf{Part II}. 

Overall, to solve the $N$-Queens Problem, a total of $(\frac{3}{2}N^2+\frac{N}{2}-1)$ qubits are used. Hence, the problem is solvable in $\mathcal{O}(N^2)$ qubits, implying that the problem is polynomial with respect to memory. The first $N^2$ qubits denote the system qubits having state $\ket{\Psi}$, which is a superposition of all basis states representing configurations satisfying the Row Criteria. The rest $(\frac{N^2}{2}+\frac{N}{2}-1)$ ancillary qubits are entangled with $\ket{\Psi}$, all of which possess $\ket{1}$ state only for the basis states encoding the solutions to the $N$-Queens Problem, following the execution of our protocol. The other basis states will have at least one associated ancilla which is in $\ket{0}$ state. Also, a total $(N+2)(N-1)$ operations are required for reducing the search space to the configurations satisfying the Column Criteria, besides the Row Criteria. It has been found that $\mathcal{O}(N^3)$ operations are required for the diagonal checks, a proof of which is presented in Supplementary Section. As a result, it takes quantum operations of order $\mathcal{O}(N^3)$ in total to solve the problem, demonstrating that the $N$-Queens Problem can be efficiently solved on a quantum computer in both polynomial time ($\mathcal{O}(N^3)$) and polynomial memory ($\mathcal{O}(N^2)$) using our protocol.

We simulated our $N$-Queens Solver protocol using the QISKit simulator, which is the same simulator used in IBM Quantum Experience as part of its Custom Topology feature. The simulation was performed for $N=4$ case, for which we know there exist $2$ solutions. Hence, a total of $25$ qubits ($16$ system qubits plus $9$ ancillas) were used.~Measurements were performed on all $25$ qubits.~The simulation was performed $310$ times.~It was observed that in two out of $183$ different measurement outcomes that were obtained, all the $9$ ancillary measurements yielded $\ket{1}$ state. The corresponding measurement reading of the $16$-qubit system in each such case represented a chessboard configuration that satisfied the $N$-Queens Criteria. Hence, we successfully obtained both the solutions to the $4$-queens problem by the implementation of our proposed algorithm. The QASM code for performing the given simulation has been presented in Supplementary Section.

\section{Application of $N$-Queens Problem: Satellite Communication}

One possible application of the $N$-Queens Problem is the transmission of information from satellites orbiting around the Earth without any data loss due to interference \cite{qnq_LutzITVT1991}. The $N$-Queens solution configurations can aid in maximizing the amount of information transmitted and the land area over which it is transmitted for a fixed number of $N$ beam reflectors on a satellite that bounds an area of $N\times N$ sq. units in space. The information is transmitted in the form of beams. 

\begin{figure}[!ht]
	\includegraphics[scale=0.45]{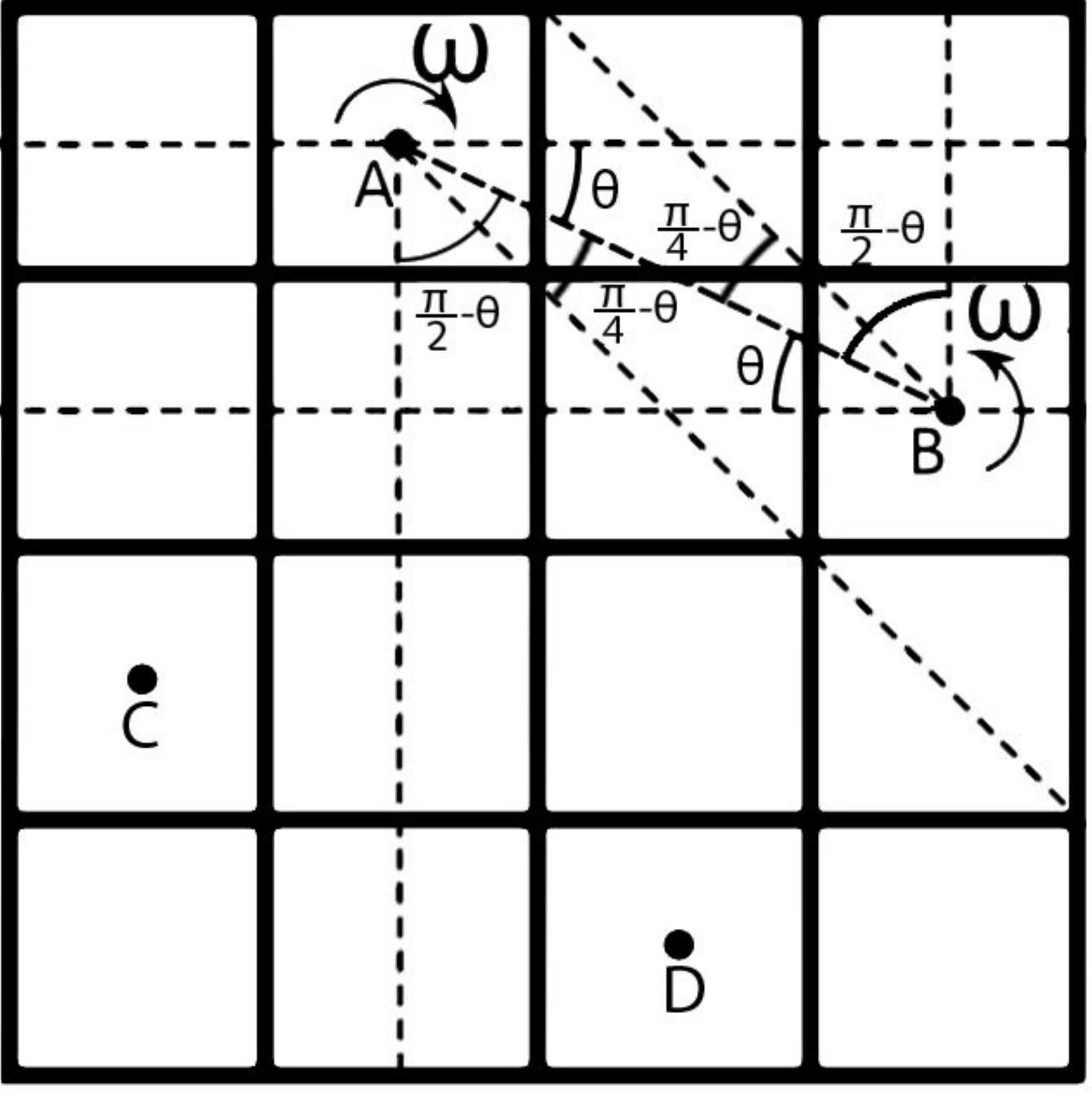}
    \caption{\textbf{Beam configuration for satellite communication.} A, B, C and D are four beam reflectors that are placed in the satellite plane (represented by the grid) in the $N$-Queens configuration. They all reflect beams horizontally, vertically and diagonally at all times (denoted by the dotted lines). A and B rotate at speeds $\omega$ and $-\omega$. If A and B rotate by $+\theta$ and $-\theta$ degrees respectively, they interfere with a phase difference that is random with time, along the line joining the two. The same happens if they rotate by $+(\pi/4-\theta)$ and $-(\pi/4-\theta)$ respectively, or $+(\pi/2-\theta)$ and $-(\pi/2-\theta)$ respectively.}. 
	\label{qnq_Fig2}
\end{figure}

Here, we make a few assumptions:
\begin{enumerate}
    \item All the N beam reflectors orbit the Earth in one plane i.e. the plane of the satellite at all times, and this plane bisects the Earth.
    \item The orientation of the plane of the satellite with respect to the Earth remains constant as the satellite revolves around the Earth.
    \item The N beam reflectors can transmit information vertically, horizontally and diagonally along all the planes.
    \item The N beam reflectors can only rotate about one of the three mutually orthogonal axes at a time i.e, two axes in the plane of the satellite and one axis perpendicular to the plane.
    \item The N beam reflectors have no restrictions while rotating about the axes that lie on the plane of the satellite. However, while rotating about the axis perpendicular to the satellite, all beam reflectors must rotate in the same direction, and with the same speed. 
    \end{enumerate}

It is obvious that if two signal waves are transmitted along the same diagonal, row or column, they interfere continuously after some point with a phase difference that varies randomly with time. This leads to the superposition of the signal waves, which results in the loss of information sent by both beam reflectors.

In an $N$-Queens solution state configuration, none of the N beams transmit information along the same row, column or diagonal. Therefore, the beams produced meet each other a maximum of once and then separate out again. This is because the waves transmitted by the beams have different frequencies. Even if they interfere at a point, they don't superpose with each other and the waveform remains the same. This results in the preservation of information.~We propose this as an optimal method for arranging beams on a satellite so as to maximize the information received by the detectors, and also the land area over which it is transmitted, without any loss of information.

\section{Discussion and Conclusion} 
We proposed a computational algorithm based on the quantum principles of superposition and entanglement to distinguish the solutions of the $N$-Queens Problem from among a superposition of $N^N$ configurations of the chessboard satisfying the criterion that each row must have only one queen. We encode the $N\times N$ chessboard matrix as a system of $N^2$ qubits, initialized in the state which is a tensor product of $N$ $W$-states, each comprising $N$ qubits, before the computation begins. Hence, the initial search space is the equal superposition of all $N^N$ configurations satisfying the Row Criteria, thus exploiting the property that quantum bits can exist in a superposition of basis states. Also, a total of $(\frac{N^2}{2}+\frac{N}{2}-1)$ ancillary qubits are utilized, which become entangled to the system and exist in different states in association with individual basis states of the initial search space (superposition state) as a consequence of our algorithm. Thus, each chessboard configuration in the initial search space is provided its signature ``marker" composed of $(\frac{N^2}{2}+\frac{N}{2}-1)$ ancillas. A specific configuration of these ancillas is bound to be associated with only the basis states encoding a chessboard configuration meeting the $N$-Queens Criteria. Specifically, $N-1$ ancillas are involved in the distinction of configurations having one queen in each row as well as in each column. In our proposed algorithm, the distinguishing feature of these configurations (basis states) is that all the entangled $N-1$ ancillas will be in $\ket{1}$ state. Likewise, the remaining $\frac{N^2-N}{2}$ ancillas are used for distinguishing configurations having one queen in every possible diagonal, from among the ones meeting the row and column criteria. Hence, the quantum advantage of superposition and entanglement manifested by a quantum computer enables quantum bits to simultaneously encode all elements of a database, and label the solution elements, depending on the problem \cite{qnq_Bravyiphysrev2016}. This is equivalent to finding the solutions to the problem, albeit the solutions cannot be made ``observable" to a classical observer unless some measurements are performed. The solutions are hidden in the quantum realm. These features enable a quantum computer to solve the NP-Complete $N$-Queens Problem in polynomial time, specifically in $\mathcal{O}(N^3)$ quantum operations using our algorithm, and also in polynomial memory, specifically in $\mathcal{O}(N^2)$ qubits. By performing several repetitions of the algorithm on a quantum computer, and noting down the list of all possible measurement outcomes, one can find out the solutions to the $N$-Queens Problem by observing the measurement outcome of the associated ancillas. 

The $N$-Queens Problem can potentially find applications in various real life situations, given that it is solvable in polynomial time and with polynomial resources. We concluded by proposing an application of the $N$-Queens Problem in satellite communication in which maximum number of beam reflectors can be arranged in a given satellite space such that information can be transmitted with no data loss due to interference. Since the beam reflectors are allowed to rotate, albeit with certain constraints, the surface area of the Earth receiving the transmissions without data loss is maximized.  

\section*{Acknowledgements} We thank Manabputra (NISER) for his fruitful discussions and Pronoy Das (IISER Kolkata) for making digital figures. We acknowledge the support of IBM Quantum Experience for providing access to the quantum processors. The views expressed are those of the authors and do not reflect the official policy or position of IBM or the IBM Quantum Experience team. We acknowledge financial support from Kishore Vaigyanik Protsahan Yojana (KVPY) and Department of Science and Technology (DST), Government of India.

\textbf{Author Contributions} D.D. proposed the idea of a quantum $N$-Queens Solver and suggested using the $N$-Queens Solver for an application. A.D., R.J. and S.J. developed the $N$-Queens Solver algorithm and performed the simulation. A.D. designed the quantum circuit. R.J. and S.J. proposed the satellite communication application. R.J., A.D., S.J., D.D. and B.K.B. contributed to the composition of the manuscript. B.K.B. supervised the project. R.J., D.D., A.D., S.J and B.K.B. have completed the project under the guidance of P.K.P.

\textbf{Author Information} The authors declare no competing financial interests. Correspondence and requests for materials should be addressed to P.K.P. (pprasanta@iiserkol.ac.in).

\textbf{Data availability.} The data that support the findings of this study are available from the corresponding author upon reasonable request.

\bibliographystyle{naturemag}

{
\renewcommand{\addcontentsline}[3]{}

}

\onecolumngrid
\section*{Supplementary Information: A Novel Quantum N-Queens Solver Algorithm and its Simulation and Application to Satellite Communication Using IBM Quantum Experience}

The QASM code of the quantum circuit used in the simulation of our protocol for $N=4$ in IBM Quantum Experience is presented below. 
\begin{lstlisting}[language=Python]
OPENQASM 2.0;
include "qelib1.inc";
qreg qr[25];
creg cr[25];

#Hadamards required for preparation of W-states in each block before the computation starts
h qr[0];
h qr[3];
h qr[4];
h qr[7];
h qr[8];
h qr[11];
h qr[12];
h qr[15];

#To initialize the ancillas required for diagonal checks in |1> state
x qr[19];
x qr[20];
x qr[21];
x qr[22];
x qr[23];
x qr[24];
x qr[0];
x qr[3];
x qr[4];
x qr[7];
x qr[8];
x qr[11];
x qr[12];
x qr[15];

#Preparing each of the 4 blocks in the 4-qubit W-state 
ccx qr[0],qr[3],qr[1];
ccx qr[4],qr[7],qr[5];
ccx qr[8],qr[11],qr[9];
ccx qr[12],qr[15],qr[13];

x qr[0];
x qr[3];
x qr[4];
x qr[7];
x qr[8];
x qr[11];
x qr[12];
x qr[15];

ccx qr[0],qr[3],qr[2];
ccx qr[4],qr[7],qr[6];
ccx qr[8],qr[11],qr[10];
ccx qr[12],qr[15],qr[14];

cx qr[2],qr[0];
cx qr[2],qr[3];
cx qr[6],qr[4];
cx qr[6],qr[7];
cx qr[10],qr[8];
cx qr[10],qr[11];
cx qr[14],qr[12];
cx qr[14],qr[15];


#Circuit for performing column check (Fig. 1 Part I) using 3 ancillary qubits
h qr[16];
h qr[17];
h qr[18];

cu1(3.14159265358979) qr[16],qr[0];
cu1(3.14159265358979) qr[16],qr[4];
cu1(3.14159265358979) qr[16],qr[8];
cu1(3.14159265358979) qr[16],qr[12];
cu1(3.14159265358979) qr[17],qr[1];
cu1(3.14159265358979) qr[17],qr[5];
cu1(3.14159265358979) qr[17],qr[9];
cu1(3.14159265358979) qr[17],qr[13];
cu1(3.14159265358979) qr[18],qr[2];
cu1(3.14159265358979) qr[18],qr[6];
cu1(3.14159265358979) qr[18],qr[10];
cu1(3.14159265358979) qr[18],qr[14];

h qr[16];
h qr[17];
h qr[18];

#Circuit for performing diagonal check (Fig. 1 Part II) using 6 ancillary qubits

#Evaluating f(Q1,Q2), value of which is stored in qr[19]
ccx qr[0],qr[5],qr[19];
ccx qr[1],qr[4],qr[19];
ccx qr[1],qr[6],qr[19];
ccx qr[2],qr[5],qr[19];
ccx qr[2],qr[7],qr[19];
ccx qr[3],qr[6],qr[19];

#Evaluating f(Q1,Q3), value of which is stored in qr[20]
ccx qr[0],qr[10],qr[20];
ccx qr[1],qr[11],qr[20];
ccx qr[2],qr[8],qr[20];
ccx qr[3],qr[9],qr[20];

#Evaluating f(Q1,Q4), value of which is stored in qr[21]
ccx qr[0],qr[15],qr[21];
ccx qr[3],qr[12],qr[21];

#Evaluating f(Q2,Q3), value of which is stored in qr[22]
ccx qr[4],qr[9],qr[22];
ccx qr[5],qr[8],qr[22];
ccx qr[5],qr[10],qr[22];
ccx qr[6],qr[9],qr[22];
ccx qr[6],qr[11],qr[22];
ccx qr[7],qr[10],qr[22];

#Evaluating f(Q2,Q4), value of which is stored in qr[23]
ccx qr[4],qr[14],qr[23];
ccx qr[5],qr[15],qr[23];
ccx qr[6],qr[12],qr[23];
ccx qr[7],qr[13],qr[23];

#Evaluating f(Q3,Q4), value of which is stored in qr[24]
ccx qr[8],qr[13],qr[24];
ccx qr[9],qr[12],qr[24];
ccx qr[9],qr[14],qr[24];
ccx qr[10],qr[13],qr[24];
ccx qr[10],qr[15],qr[24];
ccx qr[11],qr[14],qr[24];

#Performing measurement on all qubits
measure qr[0] -> cr[24];
measure qr[1] -> cr[23];
measure qr[2] -> cr[22];
measure qr[3] -> cr[21];
measure qr[4] -> cr[20];
measure qr[5] -> cr[19];
measure qr[6] -> cr[18];
measure qr[7] -> cr[17];
measure qr[8] -> cr[16];
measure qr[9] -> cr[15];
measure qr[10] -> cr[14];
measure qr[11] -> cr[13];
measure qr[12] -> cr[12];
measure qr[13] -> cr[11];
measure qr[14] -> cr[10];
measure qr[15] -> cr[9];
measure qr[16] -> cr[8];
measure qr[17] -> cr[7];
measure qr[18] -> cr[6];
measure qr[19] -> cr[5];
measure qr[20] -> cr[4];
measure qr[21] -> cr[3];
measure qr[22] -> cr[2];
measure qr[23] -> cr[1];
measure qr[24] -> cr[0];
\end{lstlisting}

In our protocol, we utilize $N-1$ ancillary qubits to find out the parity of the column sums of $N-1$ columns, i.e. every column except the last. This is because, the only case in which all such $N-1$ column sums are odd is when each of the $N$ column sums are $1$. If at least one of the $N$ column sum is greater than $1$ and is odd, then there are at least two out of the $N$ column sums which is zero. If at least one of the $N$ column sum is greater than $1$ and is even, then there is at least one out of the $N$ column sums which is zero. These claims are a result of the following proposition, which is based on the fact that the total of the column sums is constrained to be $N$, a consequence of the already satisfied Row Criteria. This proposition enables us to check the parity of the column sums of any $N-1$ columns instead of all $N$ columns.\\

\begin{proposition}
Let the sum of $N$ non-negative integers be $N$, $N\in\mathbb{N}$. Then, the number of even integers among the $N$ non-negative integers is even, where $0$ is also considered to be even.
\end{proposition}

\begin{proof}
We know that the sum of any number of even integers is even. However, the sum of an even number of odd integers is even and the sum of an odd number of odd integers is odd.\\

Among the $N$ integers, let us assume $N_o$ is the number of odd integers and $N_e$ is the number of even integers (which includes $0$). It is evident that $N_o+N_e=N$. Also, let $S_o$ and $S_e$ be the sum of the odd integers and even integers respectively. Thus, $S_o+S_e=N$. It is to be noted that $S_e$ is always even.\\

Consider the following cases,\\

\textbf{$N$ is odd:} $N$ is odd iff either $N_o$ or $N_e$ is odd. Also, since $S_e$ is even, $S_o$ is odd. This is only possible if $N_o$ is odd. Hence, $N_e$ is even.\\

\textbf{$N$ is even:} $N$ is even iff $N_o$ and $N_e$ are either both even or both odd. Since $S_e$ is even, $S_o$ is even. This is only possible if $N_o$ is even. Hence, $N_e$ is even.\\

Thus for all $N$, $N_e$ is even.
\end{proof}

Below, we show that the computational time complexity of our algorithm is $\mathcal{O}(N^3)$. 
\begin{proposition}
The time complexity for the proposed quantum $N$-Queens algorithm is of the order $\mathcal{O}(N^3)$.
\end{proposition}

\begin{proof}
The computational time complexity of an algorithm is estimated as the limiting behaviour of the number of elementary operations (gates) required by the computation as a function of the input size $N$, as $N\rightarrow\infty$. The quantum $N$-Queens algorithm performs computation for,\\

\textbf{Column checks:} $2(N-1)$ Hadamard gates are used in the column checks, since two Hadamard gates are applied to each column check ancilla. $N$ controlled phase shifts are applied to the system from each ancilla, hence a total of $N(N-1)$ controlled phase shifts are used. Thus, a total of $(N-1)(N+2)$ gates are used in the column checks.\\ 

\textbf{Diagonal checks:} The diagonal checks are performed by imagining the system of $N^2$ qubits to be arranged as a $N\times N$ matrix, such that the $i$th qubit of the $j$th block encodes the matrix element of the $i$th column and $j$th row. Every pair of qubits that encode elements positioned diagonally are checked for the diagonal condition. Since for every such pair a Toffoli gate with the corresponding qubits as controls is used, the number of gates required is equal to the number of such diagonal pairs. The general equation for the number of diagonal pairs is, $\sum_{i=1}^{N-1}\sum_{j=1}^{N-i} 2(N-j)$, which can be equivalently notated as 2$\sum_{i=1}^{N-1}\sum_{j=1}^{i} (N-j)$. This evaluates to $N^2(N-1)-N(N-1)-\frac{N(N-1)(N-2)}{3}$.\\

Thus, the total number of gates as a function of input size $N$ has the limiting behaviour (as $N\rightarrow\infty$) defined by $f:\mathbb{N}\rightarrow\mathbb{N},\text{ }f(N)=N^3$. Hence, the time complexity is of the order $\mathcal{O}(N^3)$.
\end{proof}
	
\end{document}